\documentclass[a4paper,11pt]{article}
\usepackage[american]{babel}
\usepackage[boxruled,vlined]{algorithm2e}
\usepackage{a4wide}
\usepackage{amsmath}
\usepackage{amssymb}
\usepackage{amsthm}
\usepackage{pxfonts}

\usepackage{sfmath}

\usepackage[T1]{fontenc}
\newcommand{\EDF}{\textrm{EDF}}
\newcommand{\RM}{\textrm{RM}}
\newcommand{\DM}{\textrm{DM}}

\newcommand{\drop}[1]{}
\newtheorem{Theorem}{Theorem}
\newtheorem{Definition}{Definition}
\newtheorem{Lemma}[Theorem]{Lemma}
\newtheorem{Corollary}[Theorem]{Corollary}

\newcommand{\equals}{\stackrel{\mathrm{def}}{=}}
\newcommand{\lcm}{\mathrm{lcm}}
\newcommand{\maxm}{\mathrm{max}}

\begin{document}

\title{Exact Feasibility Tests for Real-Time Scheduling of
Periodic Tasks upon Multiprocessor Platforms\footnote{This paper is
an extended version of ``Feasibility Intervals for Fixed-Priority Real-Time
Scheduling on Uniform Multiprocessors'', Proceedings of 11th IEEE
International Conference on Emerging Technologies and Factory
Automation (ETFA06) and of ``Feasibility Intervals for
Multiprocessor Fixed-Priority Scheduling of Arbitrary Deadline Periodic
Systems'', Proceedings of 10th Design, Automation and Test in Europe
(DATE07).}}

\date{}

\author{
\begin{tabular}[t]{c@{\extracolsep{7em}}c}
  Liliana Cucu\thanks{Supported in part by FNRS Grant.}  & Jo\a"el Goossens \\
  LORIA-INPL & 
Universit\'e Libre de Bruxelles (\textsc{u.l.b.})\\
615 rue du Jardin Botanique &50 Avenue Franklin D. Roosevelt\\
54600 Villers-les-Nancy, France &1050 Brussels, Belgium  \\
{\em liliana.cucu@loria.fr} & {\em joel.goossens@ulb.ac.be}
\end{tabular}
}

\maketitle

\begin{abstract}
In this paper we study the global scheduling of periodic task
systems upon multiprocessor platforms. We first show two very
general properties which are well-known for uniprocessor platforms
and which remain for multiprocessor platforms: \textit{(i)} under
few and not so restrictive assumptions, we show that feasible
schedules of periodic task systems are periodic from some point with
a period equal to the least common multiple of task periods and
\textit{(ii)} for the specific case of synchronous periodic task
systems, we show that feasible schedules repeat from the origin. We
then present our main result: we characterize, for task-level fixed-priority schedulers and for asynchronous constrained or arbitrary deadline periodic task models, \emph{upper bounds} of the first time instant where the schedule repeats. We show that job-level fixed-priority schedulers are predictable upon unrelated multiprocessor platforms. For task-level fixed-priority schedulers, based on the upper bounds and the predictability property, we provide for asynchronous constrained or arbitrary deadline periodic task sets, \emph{exact} feasibility tests. Finally, for the \emph{job}-level fixed-priority \EDF{} scheduler, for which such an upper bound remains unknown, we provide an \emph{exact} feasibility test as well.
\end{abstract}

\section{Introduction} \label{intro} The use of computers to control
safety-critical real-time functions has increased rapidly over the
past few years. As a consequence, real-time systems --- computer
systems where the correctness of each computation depends on both the
logical results of the computation and the time at which these results
are produced --- have become the focus of much study. Since the
concept of ``time'' is of such importance in real-time application
systems, and since these systems typically involve the sharing of one
or more resources among various contending processes, the concept of
scheduling is integral to real-time system design and
analysis. Scheduling theory as it pertains to a finite set of requests
for resources is a well-researched topic. However, requests in
real-time environment are often of a recurring nature. Such systems
are typically modelled as finite collections of simple, highly
repetitive tasks, each of which generates jobs in a very predictable
manner. In this work, we consider \emph{periodic task systems}, each periodic 
task $\tau_{i}$ generates jobs at each integer multiple of its period $T_i$ with the restriction that the first job is released at time $O_i$ (the task
offset).


The \emph{scheduling algorithm} determines which job[s] should be
executed at each time instant. When there is at least one schedule
satisfying all constraints of the system, the system is said to be
\emph{feasible}.

\emph{Uniprocessor} real-time systems are well studied since the
seminal paper of Liu and Layland~\cite{Liu} which introduces a model
of periodic systems. The literature considering scheduling algorithms
and feasibility tests for uniprocessor scheduling is tremendous. In
contrast for \emph{multiprocessor} parallel machines the problem of
meeting timing constraints is a relatively new research area.


In the design of scheduling algorithms for multiprocessor environments,
one can distinguish between at least two distinct approaches. In
\emph{partitioned scheduling}, all jobs generated by a task are
required to execute on the \emph{same} processor. \emph{Global scheduling},
by contrast, permits \emph{task migration} (i.e., different jobs of an
individual task may execute upon different processors) as well as \emph{job
migration} (an individual job that is preempted may resume execution
upon a processor different from the one upon which it had been
executing prior to preemption).

From theoretical and practical point of view we can distinguish between at
least three kinds of multiprocessor machines (from less
general to more general):

\begin{description}
\item[Identical parallel machines] Platforms upon which
all the processors are identical, in the sense that they have the same
computing power.

\item[Uniform parallel machines] By contrast, each processor in a
uniform parallel machine is characterized by its own computing
capacity, a job that executes on processor $\pi_i$ of computing capacity
$s_i$ for $t$ time units completes $s_i \times t$ units of execution.

\item[Unrelated parallel machines] In unrelated parallel
machines, there is an execution rate $s_{i,j}$ associated with each
job-processor pair, a job $J_i$ that executes on processor
$\pi_j$ for $t$ time units completes $s_{i,j} \times t$ units of execution. This kind of heterogeneous architectures models dedicated processors (e.g., if $s_{i,j}=0$ means that $\pi_j$ cannot serve job $J_{i}$).
\end{description}

\paragraph{Related research.} The problem of scheduling periodic task
systems on multiprocessors was originally studied
in~\cite{liu2}. Recent studies provide a better understanding of that
scheduling problem and provide first solutions. E.g., \cite{carpenter}
presents a categorization of real-time multiprocessor scheduling
problems. It is important to notice that, to the best of our
knowledge, the literature does not provide \emph{exact} feasibility
tests for global scheduling of periodic systems upon
multiprocessors. Moreover, we know that uniprocessor feasibility
results do not remain for multiprocessor scheduling. For instance the
synchronous case (i.e., considering that all tasks start their
execution synchronously) is not the worst case anymore upon
multiprocessors. Another example is the fact that the first busy
period (see~\cite{Lehoczky90} for details) does not provide a
feasibility interval upon multiprocessors (see~\cite{goossens4} for
such counter-examples). Initial results indicate that real-time
multiprocessor scheduling problems are typically not solved by
applying straightforward extensions of techniques used for solving
similar uniprocessor problems. Unfortunately, too often, researchers
use uniprocessor arguments to study multiprocessor scheduling problems
which leads to incorrect properties. This fact motivated our rigorous
and formal approach; we will present and prove correct, rigorously,
in this paper, our exact feasibility tests (and related properties).

\paragraph{This research.} In this paper we consider
preemptive global scheduling and we present exact
feasibility tests upon multiprocessors for various scheduling
policies and various periodic task models.

Our feasibility tests are based on \emph{periodicity}
properties of the schedules and on \emph{predictability}
properties of the considered schedulers. The latter
properties are not obvious because of multiprocessor
scheduling anomalies (see~\cite{Ha} for details).

More precisely, in the first part of this paper we
prove that, under few and no so restrictive assumptions,
\emph{any} feasible schedule of periodic tasks repeat from some point
in time. Then we prove that job-level fixed-priority
schedulers (e.g., \EDF{}  and \RM) are predictable upon unrelated multiprocessor platforms.

We also characterize for task-level fixed-priority schedulers and for the various periodic task models an upper bound of the first time instant where the schedule
repeats (and its period).

Lastly, we combine the periodicity and predictability
properties to provide for these various kind of periodic
task sets and various schedulers \emph{exact} feasibility tests.

\paragraph{Organization.} This paper is organized as follows. Section~\ref{model} introduces the definitions, the model of computation and our assumptions. We prove the periodicity of feasible schedules of periodic systems in Section~\ref{sectionMainPer}. In Section~\ref{sectionExactFebTest} we prove that job-level fixed-priority schedulers (e.g., \EDF{}  and \RM) are predictable upon unrelated multiprocessor platforms and we combine the periodicity and predictability
properties to provide for these various kind of periodic  task sets and various schedulers \emph{exact} feasibility tests. Lastly, we conclude in Section~\ref{conclusion}.

\section{Definitions and assumptions}\label{model}

We consider the scheduling of periodic task systems. A
system $\tau$ is composed by $n$ periodic tasks $\tau_1,
\tau_2, \ldots, \tau_n$, each task is characterized by a
period $T_i$, a relative deadline $D_i$, an execution requirement 
$C_i$ and an offset $O_i$. Such a periodic task generates an
infinite sequence of jobs, with the $k^{\text{th}}$ job
arriving at time-instant $O_{i}+ (k - 1)T_i$ ($k = 1, 2,
\ldots$), having an execution requirement of $C_{i}$ units,
and a deadline at time-instant $O_{i}+ (k-1)T_{i}+
D_{i}.$ It is important to notice that we assume in the
first part of this manuscript that each task instance of
the same task (say $\tau_{i}$) has the very same execution
requirement ($C_{i}$); we will relax this assumption in
the second part of this manuscript by showing that our
analysis is \emph{predictable}.

We will distinguish between {\it implicit deadline} systems where
$D_i=T_i, \forall i$; {\it constrained deadline} systems where $D_i
\leq T_i, \forall i$ and {\it arbitrary deadline} systems where there
is no relation between the deadlines and the periods.
Notice that arbitrary deadline systems includes
constrained deadline ones which includes the implicit deadline ones.

In some cases, we will consider the more general problem of
scheduling set of jobs, each job $J_j=(r_j,e_j,d_j)$ is characterized
by a release time $r_j$, an execution requirement $e_i$ and an absolute
 deadline $d_j$. The job $J_j$ must execute for $e_j$ time units over
the interval $[r_j,d_j)$. A job becomes {\em active} from its release
time to its completion.

A periodic system is said to be {\it synchronous} if there is an instant where
all tasks make a new request simultaneously, i.e., $\exists t, k_1,
k_2, \ldots k_n$ such that $\forall i: t=O_i+k_iT_i$ (see
\cite{goossens5} for details). Without loss of generality, we consider
$O_i= 0, \forall i$ for synchronous systems. Otherwise
the system is said to be {\it asynchronous}.

We denote by $\tau^{(i)} \equals \{\tau_1, \ldots, \tau_i\}$,
by $O_{\max} \equals \max \{O_1, O_2, \ldots, O_n \}$, by $P_i \equals
\lcm \{T_1,\ldots, T_i\}$ and $P \equals P_n$.

We consider in this paper multiprocessor platforms $\pi$
composed of $m$ unrelated processors (or one of its
particular cases: uniform and identical platforms):
$\{\pi_1,
\pi_2, \ldots, \pi_m \}$. Execution rates $s_{i,j}$ are
associated to each task-processor pair, a task $\tau_i$ that
executes on processor $\pi_j$ for $t$ time units completes
$s_{i,j} \times t$ units of execution. For each task
$\tau_i$ we assume the associated set of processors
$\pi_{n_{i,1}} > \pi_{n_{i,2}} > \cdots > \pi_{n_{i,m}} $
ordered in the decreasing order of the execution rates
relatively to the task: $s_{i,n_{i,1}} \geq s_{i,n_{i,2}}
\geq \cdots \geq s_{i,n_{i,m}}$. For identical execution
rates, the ties are broken arbitrarily, but consistently,
such that the set of processors associated to each task is
\emph{total} ordered. Consequently, the \emph{fastest} processor
relatively to task $\tau_{i}$ is $\pi_{n_{i,1}}$, i.e., the
first processor of the ordered set associated to the task.
Moreover, for a task $\tau_i$ in the following we consider
that a processor $\pi_{a}$ is {\em faster} than $\pi_{b}$ (relatively to its associated set of processors) if $\pi_{a} >\pi_{b}$ even if we have $s_{i,a} =
s_{i,b}$. For the processor-task pair $(\pi_j, \tau_i)$ if
$s_{i,j} \neq 0$ then $\pi_j$ is said to be an \emph{eligible} processor
for $\tau_i$. Notice that these concepts and definitions
can be trivially adapted to the scheduling of jobs upon
unrelated platforms. 

We consider in this paper a discrete model, i.e., the characteristics
of the tasks and the time are integers. 


We define now the notions of the state of the system and the schedule.

\begin{Definition} [State of the system $\theta(t)$] \label{defState}
  For any arbitrary deadline system $\tau = \{ \tau_1, \ldots,
  \tau_n \}$ we define the {\em state} $\theta(t)$ of the system $\tau
  $ at instant $t$ as $\theta : \mathbb{N} \rightarrow (\mathbb{Z} \times\mathbb{N}^2)^n$ with $ \theta(t) \equals (\theta_1(t), \theta_2(t), \ldots, \theta_n(t))$ where
  
$$ \theta_i(t) \equals 
\begin{cases}
  (-1,t_1,0), & \text{
\begin{minipage}[t]{11.4cm}
  if no job of task $\tau_i$ was activated before or at
$t$. In that case it remains $t_1$ time units until the first activation of $\tau_i$. (We
  have $0< t_1 \leq O_i$.);
\end{minipage}
}\\

(n_1, t_2, t_3), & \text{
\begin{minipage}[t]{11.4cm}
otherwise. In that case 
$t_2$ is the time elapsed at
    instant $t$ since the last action of the oldest active job of
    $\tau_i$. If there are $n_1 \neq 0$ active jobs of $\tau_i$ then
    $t_3$ units were already executed for the oldest active job. If $n_{1} = 0$, 
    there is no active job of $\tau_i$ at $t$, $t_{3}$ is undefined in that case. (We
    have $0 \leq n_1 \leq \lceil \frac{D_i}{T_i} \rceil $, $0 \leq t_2
    < T_i\cdot \lceil \frac{D_i}{T_i} \rceil$ and $0 \leq t_3 < C_i$.)
\end{minipage}
}
\end{cases}$$
\end{Definition}
Notice that at any instant $t$ several jobs of the same task might be
active and we consider that the oldest job is scheduled first,
i.e., the FIFO rule is used to serve the various jobs of
given task. 

\drop{
If we consider the case of constrained deadline task systems, then Definition~ref{defState} is modified as follows:
\begin{Definition} [State of the system $\theta(t)$] \label{defStatebis}
  For any constrained deadline system $\tau = \{ \tau_1, \ldots,
 \tau_n \}$ we define the {\em state} $\theta(t)$ of the system $\tau
  $ at instant $t$ as $\theta : \mathbb{N} \rightarrow (\{-1, 0,1 \}
  \times \mathbb{N}^2)^n$ with $ \theta(t) \equals (\theta_1(t),
  \theta_2(t), \ldots, \theta_n(t))$ where
  
$$ \theta_i(t) \equals 
\begin{cases}
  (-1,t_1,0), & \text{
\begin{minipage}[t]{12cm}
  if no job of task $\tau_i$ was activated before or at $t$and it remains $t_1$
time units until the first activation
    of $\tau_i$. (We have $0< t_1 \leq O_i$);
\end{minipage}}\\
  
  (0, t_2, 0), & \text{
\begin{minipage}[t]{12cm} if at least one job of $\tau_i$ was already activated
before $t$, but there is no active job of $\tau_i$ at $t$. The time elapsed
since its last activation is $t_2$. (We have $0 < t_2 < T_i$);
\end{minipage}}\\
    
    (1, t_3, t_4), & \text{
    \begin{minipage}[t]{12cm} if there is an active job of $\tau_i$
     at instant $t$ the time elapsed since
    its last action is $t_3$ and $t_4$ units were already executed. (We have $0
\leq t_3 < T_i$  and $0 \leq t_4 < C_i$.)
\end{minipage}}
\end{cases}$$
\end{Definition}
}

\begin{Definition} [Schedule $\sigma(t)$] \label{defSched}
  For any task system $\tau = \{ \tau_1, \ldots, \tau_n \}$ and any
   set of $m$ processors $\{\pi_1, \ldots, \pi_m \}$ we define the
  {\em schedule} $\sigma(t)$ of system $\tau $ at instant $t$ as
  $\sigma : \mathbb{N} \rightarrow \{0, 1, \ldots, n \}^m$  where $ \sigma(t) \equals ( \sigma_1(t),
  \sigma_2(t), \ldots,
  \sigma_m(t) )$ with \\
  $\sigma_j(t) \equals \left\{
\begin{array}{ll}
0, & \text{if there is no task scheduled on } \pi_j \\ &\text{at instant } t; \\
i, & \text{if task } \tau_i \mbox{ is scheduled on } \pi_j
\text{ at instant } t.  
\end{array}
\right. \forall 1 \leq j \leq m. 
$
\end{Definition}

Notice that Definition~\ref{defSched} can be extended trivially to the
scheduling of jobs.

A system $\tau$ is said to be {\it feasible} upon a multiprocessor
platform if there exists at least one schedule in which all tasks meet
their deadlines. If $A$ is an algorithm which schedules $\tau$ upon a
multiprocessor platform to meet its deadlines, then the system $\tau$
is said to be $A$-feasible.

In this work, we consider that \emph{task parallelism is forbidden}: a task
cannot be scheduled at the same instant on different processors,
i.e. $\nexists j_1 \neq j_2 \in \{1, 2, \ldots, m \}$ and $t \in
\mathbb{N}$ such that $\sigma_{j_1}(t)= \sigma_{j_2}(t) \neq 0$.

The scheduling algorithms considered in this paper are {\em
  deterministic} and work-conserving with the following definitions

\begin{Definition}[Deterministic algorithm]\label{detAlg} 
A scheduling algorithm is said to
be \emph{deterministic} if it generates a unique schedule for any
given sets of jobs .
\end{Definition}

In uniprocessor (or identical multiprocessor) scheduling, a {\em
work-conserving} algorithm is defined to be the one that never idles a
processor while there is at least one active task. For unrelated multiprocessors
we adopt the following definition:

\begin{Definition} [Work-conserving algorithm] \label{defWorkC} 
  An unrelated multiprocessor scheduling algorithm is said
\emph{work-conserving} if at each instant, the algorithm schedules jobs
  to processors as follows: the highest priority (active) job $J_i$ is
  scheduled on its fastest (and eligible) processor $\pi_j$. The very
  same rule is then applied to the remaining active jobs on the
  remaining available processors.
\end{Definition}

Moreover, we will assume that the decision of the scheduling
algorithm at time $t$ is not based on the past, nor on the
actual time $t$ but only on the characteristics of active tasks and on
the state of the system at time $t$. More formally, we consider
\emph{memoryless} schedulers.

\begin{Definition}[Memoryless algorithm]\label{def:memoryless} 
  A scheduling algorithm is said to be \emph{memoryless} if the
  scheduling decision made by it at time $t$ depends only on the
  characteristics of active tasks and on the current state of the
  system, i.e., on $\theta(t)$.
\end{Definition}

Consequently, for memoryless and deterministic schedulers we have the
following property:
\[
\forall t_{1}, t_{2} \;\text{such that}\; \theta(t_{1})=\theta(t_{2})
\;\text{then}\; \sigma(t_{1})=\sigma(t_{2}).
\]

It follows by Definition~\ref{defWorkC} that a processor $\pi_j$ can
be idled and a job $J_i$ can be active at the same time if and
only if $s_{i,j}=0$.

In the following, we will distinguish between two kinds of scheduler:

\begin{Definition}[Task-level fixed-priority]\label{task-level}
The priorities are assigned to the tasks
beforehand, at run-time each job \emph{inherits} of its task
priority and remains constant.
\end{Definition} 

\begin{Definition}[Job-level fixed-priority]\label{prioDr}
  A scheduling algorithm is a {\em job-level fixed-priority}
  algorithm if and only if it satisfies the condition that for every
  pair of jobs $J_i$ and $J_j$, if $J_i$ has higher priority than
  $J_j$ at some time instant, then $J_i$ always has higher priority
  than $J_j$.
\end{Definition} 

Popular task-level fixed-priority schedulers include the Rate Monotonic (\RM) or the Deadline Monotonic (\DM); popular job-level fixed-priority schedulers include the Earliest Deadline First (\EDF), see~\cite{Liu} for details.

We denote by $\delta_i^k$ the $k^{\text{th}}$  job of task $\tau_i$ which
becomes active at time instant $R_i^k \equals O_i+(k-1)T_i$.

\begin{Definition}[$\epsilon_{i}^{k}(t)$]
  For any task $\tau_i$, we define $\epsilon_{i}^{k}(t)$ to be the amount of
  time already executed for $\delta_i^k$ in the interval $[R_i^k,
  t)$.
\end{Definition}

We introduce now the availability of the processors for any schedule
$\sigma(t)$.

\begin{Definition} [Availability of the processors $a(t)$, task scheduling]
\label{defAvai} For any task system $\tau = \{
  \tau_1, \ldots, \tau_n \}$ and any set of $m$ processors
  $\{\pi_1, \ldots, \pi_m \}$ we define the {\em availability of the processors}
  $a(t)$ of system $\tau$ at instant $t$ as the set of available processors
   $a(t)  \equals \{\/j \mid \mbox{ } \sigma_j(t)=0 \} \subseteq \{1,
  \ldots, m\}$.
\end{Definition}

\section{Periodicity of feasible schedules}\label{sectionMainPer}

It is important to remind that we assume in this section that all task execution requirements are \emph{constant}, we will relax this assumption in Section~\ref{sectionExactFebTest}. 

This section contains four parts,  we give in each part of this section results concerning the periodicity of feasible schedules. By periodicity (assuming that the period is $\gamma$) of a schedule $\sigma$, we understand there is a time instant $t_{0}$ such that $\sigma(t)=\sigma(t+\gamma), \forall t \geq t_{0}$.

The first part of this section provides periodicity results for a
(very) general scheduling algorithm class: deterministic,
memoryless and work-conserving schedulers. 

The second part of this section provides periodicity results for
synchronous periodic task systems.

The third and the fourth part of this section present periodicity
results for task-level fixed-priority scheduling algorithms for constrained and arbitrary deadline systems, respectively. 

\subsection{Periodicity of deterministic, memoryless and
  work-conserving scheduling algorithms} \label{genGenSyst}

We show that feasible schedules of periodic task systems obtained
using deterministic, memoryless and work-conserving algorithms are
periodic from some point. Moreover we prove that the schedule repeats
with a period equal to $P$ for a sub-class of such schedulers. Based on that property, we provide two
interesting corollaries for preemptive task-level fixed-priority algorithms (Corollary~\ref{task-levelperiod}) and for preemptive deterministic $\EDF$\footnote{by deterministic {$\EDF$} we mean that ambiguous situations are solved deterministically.} (Corollary~\ref{edfAll1}).

We present first two preliminary results in order to prove Theorem~\ref{perGen}.
 
\begin{Lemma}\label{thZero} For any deterministic and memoryless
  algorithm $A$, if an asynchronous arbitrary deadline system $\tau$
  is $A$-feasible, then the $A$-feasible schedule of $\tau$ on $m$
  unrelated processors is periodic with a period \emph{divisible} by
  $P$. 
\end{Lemma}

\begin{proof}
  First notice that from $t_0 \geq O_{\max}$ all tasks are released,
  and the configuration $\theta_i(t)$ of each task is a triple of
  finite integers $(\alpha, \beta, \gamma)$ with $ \alpha \in \{0,1, \ldots, \lceil \frac{D_i}{T_i} \rceil \}$, $0 \leq \beta < \max_{1 \leq i
    \leq n}T_i$ and $0 \leq \gamma < \max_{1 \leq i \leq
    n}C_i$. Therefore there is a finite number of different system
  states, hence we can find two distinct instants $t_1$ and $t_2$
  ($t_{2} > t_1 \geq t_0)$ with the same state of the system
  ($\theta(t_1)=\theta(t_2)$). The schedule repeats from that instant
  with a period dividing $t_2-t_1$, since the scheduler is
  deterministic and memoryless.
  
  Notice that since the tasks are periodic, the arrival pattern of jobs repeats with a period equal to $P$ from $O_{\maxm}$.

  We prove now by contradiction that $t_2-t_1$ is necessarily a multiple
  of $P$. We suppose that $\exists k_1 < k_2 \in \mathbb{N}$ such
  that $t_i= O_{\maxm}+k_iP+ \Delta_i, \forall i \in \{1,2 \}$ with
  $\Delta_1 \neq \Delta_2$, $\Delta_1, \Delta_2 \in [0,P)$ and $\theta(t_1)=\theta(t_2)$. This implies that there are tasks for which the time elapsed since the last activation at $t_1$ and the time elapsed since the last
  activation at $t_2$ are not equal. But this is contradiction with the fact that $\theta(t_1)=\theta(t_2)$. Consequently $\Delta_1$ must be equal to
  $\Delta_2$ and, thus, we have $t_2-t_1 = (k_2 - k_1)P$. 
\end{proof}

For a sub-class of schedulers, we will show that the period of the schedule is $P$, but first a definition (inspired from~\cite{GD99b}):

\begin{Definition}[Request-dependent scheduler]
A scheduler is said to be \emph{request-dependent} if $\forall i,j,k,\ell, t: \delta_{i}^{k+h_{i}}(t+P) > \delta_{j}^{\ell + h_{j}}(t+P)$ if and  only if $\delta_{i}^{k}(t) > \delta_{j}^{\ell }(t)$, where $\delta_{i}^{k}(t) > \delta_{j}^{\ell }(t)$ means that the request $\delta_{i}^{k}(t)$ has a higher priority than the request $\delta_{j}^{\ell }(t)$.
\end{Definition}

The next lemma extend results given for arbitrary
deadline task systems in the \emph{uniprocessor} case
(see~\cite{thesisJG}, p.~$55$ for details).

\begin{Lemma} \label{prepTh} For any preemptive, job-level fixed-priority and request-dependent algorithm $A$ and any asynchronous arbitrary deadline system $\tau$ on $m$ unrelated processors, we have that:
  for each task $\tau_i$, for any time instant $t \geq O_i$ and $k$
  such that $R_i^k \leq t \leq R_i^k+D_i$, if there is no deadline
  missed up to time $t+P$, then $\epsilon_{i}^{k}(t) \geq
  \epsilon_{i}^{k+h_i}(t+P)$ with $h_i \equals \frac{P}{T_i}$.
\end{Lemma}

\begin{proof}
  The proof is made by contradiction. Notice first that the function
  $\epsilon_i^k(\,)$ is a non-decreasing discrete step function with $0
  \leq \epsilon_i^k(t) \leq C_i, \forall t$ and $\epsilon_i^k(R_i^k)=0=
  \epsilon_i^{k+h_i}(R_i^{k+h_i}), \forall k$.

  We assume that a \emph{first} time instant $t$ exists such that there are
  $j$ and $k$ with $R_j^k \leq t \leq R_j^k+D_j$ and
  $\epsilon_{j}^{k}(t) < \epsilon_{j}^{k+h_j}(t+P)$. This assumption
  implies that there is a time instant $t'$ with $R_j^{k} \leq t' < t$
  such that $\delta_j^{k+h_j}$ is scheduled at $t'+P$ while
  $\delta_j^k$ is not scheduled at $t'$. We obtain that $m$ higher priority jobs are scheduled at $t'$ and among these jobs there is, at least, one
  job $\delta_{{\ell}}^{k_{\ell}+h_{\ell}}$ of a task $\tau_{{\ell}}$ with
  ${\ell} \in \{1, 2, \ldots, n \}$ that is not
  scheduled at $t'+P$, while $\delta_{{\ell}}^{k_{\ell}}$
  is scheduled at $t'$ ($h_{\ell}\equals\frac{P}{T_{\ell}}$). This implies
  that $\epsilon_{{\ell}}^{k_{\ell}}(t') <
  \epsilon_{{\ell}}^{k_{\ell}+h_{\ell}}(t'+P)=C_{{\ell}}$
  but this is a contradiction with the fact that $t$ is the first such
  time instant.
\end{proof}

\begin{Theorem}\label{perGen} 
For any preemptive job-level fixed-priority and request-dependent algorithm $A$ and any $A$-feasible asynchronous arbitrary deadline system $\tau$ 
upon $m$ unrelated processors the schedule is periodic with a period equal to $P$. 
\end{Theorem}

\begin{proof}
  By Lemma~\ref{thZero} we have that $\exists t_i= O_{\maxm}+k_iP+d,
  \forall i \in \{1,2 \}$ with $0 \leq d < P$ such that
  $\theta(t_1)=\theta(t_2)$. We know also that the arrivals of jobs of
  tasks repeat with a period equal to $P$ from $O_{\maxm}$. Therefore
  for all time instants $t_1+kP$, $\forall k < k_2-k_1$ (i.e. $t_1+kP<
  t_2$), we have that the time elapsed since the last activation at
  $t_1+kP$ is the same for all tasks. Moreover
  since $\theta(t_1)=\theta(t_2)$ we have that
  $\epsilon_{i}^{{\ell}_i}(t_1)=\epsilon_{i}^{{\ell}_i+\frac{(k_2-k_1)P}{T_i}}(t_2)$
  with ${\ell}_i =\lceil \frac{O_{\maxm}+d}{T_i} \rceil +
  \frac{k_1P}{T_i}$,$\forall i$. But by Lemma~\ref{prepTh} we also
  have that $\epsilon_{i}^{{\ell}_i}(t_1) \leq
  \epsilon_{i}^{{\ell}_i+\frac{P}{T_i}}(t_1+P) \leq \cdots \leq
  \epsilon_{i}^{{\ell}_i+\frac{(k_2-k_1)P}{T_i}}(t_2), \forall
  i$. Consequently we obtain that $\theta(t_1) \leq \theta(t_1+P) \leq
  \cdots \leq \theta(t_2)$ and $\theta(t_1)=\theta(t_2)$ which implies
  that $\theta(t_1) = \theta(t_1+P) = \cdots = \theta(t_2)$. 
\end{proof}

\begin{Corollary}\label{task-levelperiod}
For any preemptive task-level fixed-priority algorithm $A$, if an asynchronous arbitrary deadline system $\tau$ is $A$-feasible upon $m$ unrelated processors is periodic with a period equal to $P$.
\end{Corollary}

\begin{proof}
The result is a direct consequence of Theorem~\ref{perGen}, since task-level fixed-priority algorithms are job-level fixed-priority and request-dependent schedulers.
\end{proof}

\begin{Corollary}\label{edfAll1} A feasible schedule obtained using
  deterministic request-dependent global {$\EDF$} on $m$ unrelated processors of an
  asynchronous arbitrary deadline system $\tau$ is periodic with a
  period equal to $P$. 
\end{Corollary}

\begin{proof}
  The result is a direct consequence of Theorem~\ref{perGen}, since
  {$\EDF$} is a job-level fixed-priority scheduler.
\end{proof}

\subsection{The particular case of synchronous periodic systems}
\label{SectSynchEns}

In this section we deal with the periodicity of feasible schedules of
\emph{synchronous} periodic systems. Using the results obtained for
deterministic, memoryless and work-conserving algorithms we prove in
Section \ref{synSect} that the feasible schedules of synchronous
constrained deadline periodic systems are periodic from time instant
equal to $0$. In Section \ref{synSectarb} we study arbitrary deadline
periodic systems and the periodicity of feasible schedules of these
systems using preemptive task-level fixed-priority scheduling
algorithms.

\subsubsection{Synchronous constrained deadline periodic
  systems} \label{synSect} 

In this section we deal with the particular case of synchronous
periodic task systems and we show the periodicity of feasible
schedules.

\begin{Theorem}\label{synPer}
  For any deterministic, memoryless and work-conserving  algorithm $A$,
  if a synchronous constrained deadline system $\tau$ is $A$-feasible,
  then the $A$-feasible schedule of $\tau$ on $m$ unrelated processors
  is periodic with a period $P$ that begins at instant $0$.
\end{Theorem}

\begin{proof} 
  Since $\tau$ is a synchronous periodic system, all tasks become
  active at instants $0$ and $P$. Moreover, since $\tau$ is a
  $A$-feasible constrained deadline system, all jobs occurred strictly
  before instant $P$ have finished their execution before or at
  instant $P$. Consequently, at instants $0$ and $P$ the system is in
  the same state, i.e. $\theta(0)=\theta(P)$, and a deterministic and
  memoryless scheduling algorithm will make the same scheduling
  decision. The schedule repeats with a period equal to $P$.
\end{proof}

An interesting particular case of Theorem~\ref{synPer} is the following:

\begin{Corollary}\label{edfAll2} 
  A feasible schedule obtained using deterministic global {$\EDF$} of a
  synchronous constrained deadline system $\tau$ on $m$ identical or
  unrelated processors is periodic with a period $P$ that begins at
  instant $0$. 
\end{Corollary}

\subsubsection{Synchronous arbitrary deadline periodic
  systems} \label{synSectarb}

In this section we deal with the particular case of synchronous
arbitrary deadlines task systems and we show the periodicity of
feasible schedules obtained using preemptive task-level fixed-priority
scheduling algorithms.

In the following, and without loss of generality, we consider the tasks
ordered in decreasing order of their priorities $\tau_1 > \tau_2 >
\cdots > \tau_n$.

\begin{Lemma}\label{synPerLem} For any preemptive task-level fixed-priority
  algorithm $A$ and for any synchronous arbitrary deadline system $\tau$
  on $m$ unrelated processors, if no deadline is missed in the time interval
  $[0,P)$ and if $\theta(0)=\theta(P)$, then the schedule of $\tau$ is periodic with a period $P$ that begins at instant $0$.
\end{Lemma}

\begin{proof}
  Since at time instants $0$ and $P$ the system is in the same state,
  i.e. $\theta(0) =\theta(P)$, then at time instants $0$ and $P$ a
  preemptive task-level fixed-priority algorithm will make the same scheduling
  decision and the scheduled repeats from $0$ with a period equal to
  $P$.
\end{proof}

\begin{Theorem} \label{synPerNotFeb}
For any preemptive task-level fixed-priority algorithm $A$ and any synchronous
arbitrary deadline system $\tau$ on $m$ unrelated processors, if all deadlines are met in $[0,P)$ and $\theta(0) \neq \theta(P)$, then $\tau$ is not $A$-feasible.
\end{Theorem}

\begin{proof} 
In the following, we denote by $\sigma^{(i)}$ the schedule of the task subset $\tau^{(i)}$. Since $\theta(0) \neq \theta(P)$, there is more than one active job of the same task at $P$. We define $\ell \in \{1, 2, \ldots, n \}$ to be the smallest task index such that $\tau_{\ell}$ has at least two active jobs at $P_{\ell}$. In order to prove the property we will prove that $\tau_{\ell}$ will miss a deadline.

By definition of $\ell$ we have that $\theta(0) = \theta(P_{\ell-1})$ (at least for the schedule $\sigma^{(\ell-1)}$) and
by Lemma~\ref{synPerLem} we have that the time instants, such that at
least one processor is available, are periodic with a period
$P_{\ell-1}$, i.e., the schedule $\sigma^{(\ell-1)}$ obtained by
considering only the task subset $\tau^{(\ell-1)}$ is periodic with
a period $P_{\ell-1}$. Moreover, since $P_{\ell}$ is a multiple of
$P_{\ell-1}$, we know that the schedule $\sigma^{(\ell-1)}$ is periodic with
a period $P_{\ell}$. Therefore in each time interval $[ k \cdot
P_{\ell},(k+1) P_{\ell})$ with $k \geq 0$ after scheduling $\tau_1,
\tau_2, \ldots, \tau_{\ell -1}$ there is the same number $t_{\ell}$ of
time instants such that at least one processor is available and where
$\tau_{\ell}$ is scheduled. At time instant $P_{\ell}$, since the task
parallelism is forbidden, there are $\frac{P_{\ell}}{T_{\ell}} C_{\ell}-t_{\ell}$
remaining units for execution of $\tau_{\ell}$ and, consequently,
at each time instant $(k+1) \cdot P_{\ell}$ there will be $k \cdot
(\frac{P_{\ell}}{T_{\ell}} C_{\ell}-t_{\ell})$ remaining units for
execution of $\tau_{\ell}$. Consequently we can find $k_{\ell}= \left\lceil
\frac {D_{\ell}}{{P_{\ell}}/{T_{\ell}} \cdot (C_{\ell}-t_{\ell})} \right\rceil$ such that the job actived at $(k_{\ell}+1)P_{\ell}$ will miss its deadline
since it cannot be scheduled before older jobs of $\tau_{\ell}$ and
there are $k_{\ell}({P_{\ell}}/{T_{\ell}} \cdot (C_{\ell}-t_{\ell})) \geq
D_{\ell}$ remaining units for execution of $\tau_{\ell}$ at
$(k_{\ell}+1)P_{\ell}$.

Since we consider task-level fixed-priority scheduling, then the
tasks $\tau_i$ with $i > \ell$ will not interfere with the higher
priority tasks already scheduled, particularly with $\tau_{\ell}$
that misses its deadline, and consequently the system is not
$A$-feasible.
\end{proof}

\begin{Corollary}\label{synPerArbi}
For any preemptive task-level fixed-priority algorithm $A$ and any synchronous
arbitrary deadline system $\tau$ on $m$ unrelated processors, if $\tau$ is $A$-feasible, then the schedule of $A$ is periodic with a period $P$ that begins at instant $0$.
\end{Corollary}

\begin{proof}
Since $\tau$ is $A$-feasible, we know by Theorem~\ref{synPerNotFeb} that $\theta(0)=\theta(P)$. Moreover, a deterministic and memoryless scheduling algorithm will make the same scheduling decision at those instants. Consequently, the schedule repeats from the origin with a period of $P$.
\end{proof}

\subsection{Task-level fixed-priority scheduling of asynchronous constrained
  deadline systems} \label{asynSect}

In this section we give another important result: any feasible schedules on $m$ unrelated processors of asynchronous constrained deadline systems,
obtained using preemptive task-level fixed-priority algorithms, are periodic from some point (Theorem~\ref{asynPer}) and we characterize that point.

Without loss of generality we consider the tasks ordered in decreasing
order of their priorities $\tau_1 > \tau_2 > \cdots > \tau_n$.

\begin{Theorem} \label{asynPer} For any preemptive task-level fixed-priority
algorithm $A$ and any $A$-feasible asynchronous constrained deadline system $\tau$ upon $m$ unrelated processors is periodic with a period $P$ from instant $S_n$ where $S_i$ is defined inductively as follows:

  \begin{itemize}
  \item $S_1 \equals O_1$; 
  \item $S_i \equals \max \{ O_i, O_i+ \lceil \frac{S_{i-1}-O_i}{T_i}
    \rceil T_i \}, \forall i \in \{2,3, \ldots, n \}$.
  \end{itemize}
\end{Theorem}

\begin{proof}
  The proof is made by induction by $n$ (the number of tasks). We
  denote by $\sigma^{(i)}$ the schedule obtained by considering only
  the task subset $\tau^{(i)}$, the first higher priority $i$ tasks
  $\{\tau_1, \ldots, \tau_i \}$, and by $a^{(i)}$ the corresponding
  availability of the processors. Our inductive hypothesis is the
  following: the schedule $\sigma^{(k)}$ is periodic from $S_k$ with a
  period $P_k$ for all $1 \leq k \leq i$.

  The property is true in the base case: $\sigma^{(1)}$ is periodic
  from $S_1=O_1$ with period $P_1$, for $\tau^{(1)}= \{\tau_1 \}$:
  since we consider constrained deadline systems, at instant $P_1=T_1$ the
  previous request of $\tau_1$ has finished its execution and the
  schedule repeats.

  We will now show that any $A$-feasible schedules of $\tau^{(i+1)}$ are
  periodic with period $P_{i+1}$ from $S_{i+1}$.

  Since $\sigma^{(i)}$ is periodic with a period $P_{i}$ from $S_{i}$ the
  following equation is verified:

\begin{equation}
  \label{stateInter}
\sigma^{(i)}(t)=\sigma^{(i)}(t+P_i), \forall t \geq S_{i}.
\end{equation}

We denote by $S_{i+1} \equals \max \{ O_{i+1}, O_{i+1}+ \lceil
\frac{S_{i}-O_{i+1}}{T_{i+1}} \rceil T_{i+1} \}$ the first request of
$\tau_{i+1}$ not before $S_i$.

Since the tasks in $\tau^{(i)}$ have higher priority than
$\tau_{i+1}$, then the scheduling of $\tau_{i+1}$ will not interfere with
higher priority tasks which are already scheduled. Therefore, we may
build $\sigma^{(i+1)}$ from $\sigma^{(i)}$ such that the tasks
$\tau_1, \tau_2, \ldots, \tau_i$ are scheduled at the very same
instants and on the very same processors as they were in
$\sigma^{(i)}$. We apply now the induction step: for all $t \geq
S_{i}$ in $\sigma^{(i)}$ we have $a^{(i)}(t) = a^{(i)}(t +P_i)$ the
availability of the processors repeats. Notice that at those instants $t$
and $t+P_i$ the available processors (if any) are the same. Consequently, at
only these instants task $\tau_{i+1}$ {\em may} be executed.

The instants $t$ with $S_{i+1} \leq t < S_{i+1}+P_{i+1}$, where
$\tau_{i+1}$ may be executed in $\sigma^{(i+1)}$, are periodic with
period $P_{i+1} = \lcm\{P_{i},T_{i+1}\}$. Moreover, since the system is feasible and we consider constrained deadlines, the only active request of $\tau_{i+1}$ at $S_{i+1}$ (respectively at $S_{i+1}+P_{i+1}$) is the one activated at $S_{i+1}$ (respectively at $S_{i+1}+P_{i+1}$). Consequently, the instants at which the task-level fixed-priority algorithm $A$ schedules $\tau_{i+1}$ are periodic with period $P_{i+1}$. Therefore the schedule $\sigma^{(i+1)}$ repeats from $S_{i+1}$ with period equal to $P_{i+1}$ and the property is true for all $1 \leq k \leq n$, in particular for $k=n:$ $\sigma^{(n)}$ is periodic with period equal to $P$ from $S_n$ and the property follows.
\end{proof}

\subsection{Task-level fixed-priority scheduling of asynchronous arbitrary 
  deadline systems} \label{asynSectArb}

In this section we present another important result: any feasible schedule on
$m$ unrelated processors of asynchronous arbitrary deadline systems,
obtained using preemptive task-level fixed-priority algorithms, is periodic from
some point (Theorem~\ref{asynPerbis}).

\begin{Corollary}\label{prepThbis} 
  For any preemptive task-level fixed-priority algorithm $A$ and any asynchronous
  arbitrary deadline system $\tau$ on $m$ unrelated processors, we
  have that: for each task $\tau_i$, for any time instant $t \geq O_i$
  and $k$ such that $R_i^k \leq t \leq R_i^k+D_i$, if there is no
  deadline missed up to time $t+P$, then $\epsilon_{i}^{k}(t) \geq
  \epsilon_{i}^{k+h_i}(t+P)$ with $h_i \equals \frac{P}{T_i}$.
\end{Corollary}

\begin{proof}
  This result is direct consequence of Lemma~\ref{prepTh} since preemptive
  task-level fixed-priority algorithms are job-level fixed-priority and request-dependent schedulers.
\end{proof}



\begin{Corollary} \label{Coreither}
  For any preemptive task-level fixed-priority algorithm $A$ and any asynchronous
  arbitrary deadline system $\tau$ on $m$ unrelated processors, we
  have that: for each task $\tau_i$, for any time instant $t \geq
  O_i$, if there is no deadline missed up to time $t+P$, then either
  $(\alpha_i(t) < \alpha_i(t+P))$ or $[\alpha_i(t)= \alpha_i(t+P)$ and
  $\gamma_i(t) \geq \gamma_i(t+P)]$, where by the triple 
  $(\alpha_i(t), \beta_i(t), \gamma_i(t))$ we denoted $\theta_i(t)$.
\end{Corollary}

\begin{proof} If $\alpha_i(t)=0$, then either $\alpha_i(t+P) > 0$ or
  $\alpha_i(t+P)=0=\beta_i(t+P)= \beta_i(t)$. Otherwise,
  $\alpha_i(t)=n_i(t) - m_i(t)$ where $n_i(t)$ is the number of jobs
  actived before or at $t$, and $m_i(t)$ is the number of jobs that
  have completed their execution before or at $t$. We have
  $n_i(t+P)=n_i(t)+ \frac{P}{T_i}$ and by Corollary~\ref{prepThbis} we
  obtain that $m_i(t+P) \leq m_i(t) + \frac{P}{T_i}$. Consequently
  $\alpha_i(t+P)\geq \alpha_i(t)$, and if $\alpha_i(t)=\alpha_i(t+P)$
  then $m_i(t+P)= m_i(t)+ \frac{P}{T_i}$, and $\beta_i(t)=
  \epsilon^{m_i(t)+1} \geq
  \epsilon_i^{m_i(t)+1+\frac{P}{T_i}}(t+P)=\beta_i(t+P)$.
\end{proof}


\begin{Theorem} \label{asynPerbis} For any preemptive task-level fixed-priority
  algorithm $A$ and any $A$-feasible asynchronous arbitrary deadline system $\tau$
  upon $m$ unrelated processors is periodic with a period $P$ from instant $\widehat{S}_n$ where $\widehat{S}_n$
  are defined inductively as follows: 

  \begin{itemize}
  \item $\widehat{S}_1 \equals O_1$
  \item $\widehat{S}_i \equals \max \{ O_i, O_i+ \lceil
    \frac{\widehat{S}_{i-1}-O_i}{T_i} \rceil T_i \} + P_i, \qquad (i >
    1$)
   \end{itemize}
\end{Theorem}

\begin{proof}
  The proof is made by induction by $n$ (the number of tasks). We
  denote by $\sigma^{(i)}$ the schedule obtained by considering only
  the task subset $\tau^{(i)}$, the first higher priority $i$ tasks
  $\{\tau_1, \ldots, \tau_i \}$, and by $a^{(i)}$ the corresponding
  availability of the processors. Our inductive hypothesis is the
  following: the schedule $\sigma^{(k)}$ is periodic from $\widehat{S}_k$
  with a period $P_k$, for all $1 \leq k \leq i$.

  The property is true in the base case: $\sigma^{(1)}$ is periodic
  from $\widehat{S}_1=O_1$ with period $P_1=T_{1}$, for $\tau^{(1)}= \{\tau_1 \}$:
  since we consider feasible systems, at instant $P_1+O_1=T_1+O_1$ the
  previous job of $\tau_1$ has finished its execution ($C_1 \leq
  T_1$) and the schedule repeats.

  We will now show that any $A$-feasible schedule of $\tau^{(i+1)}$
  is periodic with period $P_{i+1}$ from $\widehat{S}_{i+1}$.

  Since $\sigma^{(i)}$ is periodic with a period $P_{i}$ from $\widehat{S}_{i}$ the
  following equation is verified:

\begin{equation}
  \label{stateInterbis}
\sigma^{(i)}(t)=\sigma^{(i)}(t+P_i), \forall t \geq \widehat{S}_{i}.
\end{equation}

We denote by $\widehat{S}_{i+1} \equals \max \{ O_{i+1}, O_{i+1}+ \lceil
\frac{\widehat{S}_{i}-O_{i+1}}{T_{i+1}} \rceil T_{i+1} \} +P_{i+1}$ the time
instant obtained by adding $P_{i+1}$ to the time instant which
corresponds to the first activation of $\tau_{i+1}$ after $\widehat{S}_i$.

Since the tasks in $\tau^{(i)}$ have higher priority than
$\tau_{i+1}$, then the scheduling of $\tau_{i+1}$ will not interfere
with higher priority tasks which are already scheduled. Therefore, we
may build $\sigma^{(i+1)}$ from $\sigma^{(i)}$ such that the tasks
$\tau_1, \tau_2, \ldots, \tau_i$ are scheduled at the very same
instants and on the very same processors as there were in
$\sigma^{(i)}$. We apply now the induction step: for all $t \geq
\widehat{S}_{i}$ in $\sigma^{(i)}$ we have $a^{(i)}(t) = a^{(i)}(t +P_i)$ the
availability of the processors repeats. Notice that at the instants
$t$ and $t+P_i$ the available processors (if any) are the same. Hence
at only these instants task $\tau_{i+1}$ {\em may} be executed in the
time interval $[\widehat{S}_{i+1}, \widehat{S}_{i+1}+P_{i+1})$. 

The instants $t$ such that $\widehat{S}_{i+1} \leq t < \widehat{S}_{i+1}+P_{i+1}$, where
$\tau_{i+1}$ may be executed in $\sigma^{(i+1)}$, are periodic with
period $P_{i+1}$, since $P_{i+1}$ is a multiple of $P_i$ and $\widehat{S}_{i+1}
\geq \widehat{S}_i$. We prove now by contradiction that the system is in the
same state at time instant $\widehat{S}_{i+1}$ and $\widehat{S}_{i+1}+P_{i+1}$. We suppose that  $\theta(\widehat{S}_{i+1}) \neq \theta(\widehat{S}_{i+1}+P_{i+1})$.

We first prove that $ \nexists t \in [\widehat{S}_{i+1}, \widehat{S}_{i+1}+P_{i+1})$ such
that at $t$ there is at least one available processor in
$\sigma^{(i)}$ and no job of $\tau_{i+1}$ is scheduled at $t$ in
$\sigma^{(i+1)}$. If there is such an instant $t'$, then by
Corollary~\ref{Coreither} we have that $\theta(t'-P_{i+1})=
\theta(t')$ since from the inductive hypothesis (notice that
$P_{i+1}$ is multiple of $P_i$) and since $t'-P_{i+1} \geq \widehat{S}_{i+1}-P_{i+1}
\geq \widehat{S}_{i} \geq \cdots \geq \widehat{S}_1$ we obtain that
$\theta_k(t'-P_{i+1})=\theta_k(t')$ for $1 \leq k \leq
i$. Consequently, $\theta(\widehat{S}_{i+1})= \theta(\widehat{S}_{i+1}+P_{i+1})$ which is
in contradiction with our assumption.

Secondly, since $\theta_{i+1}(\widehat{S}_{i+1}) \neq
\theta_{i+1}(\widehat{S}_{i+1}+P_{i+1})$ then by Corollary~\ref{Coreither} we
have that either there are less active jobs at $\widehat{S}_{i+1}$ than at
$\widehat{S}_{i+1}+P_{i+1}$, or if there is the same number of active jobs of
$\widehat{S}_{i+1}$ then the oldest active job at $\widehat{S}_{i+1}$ was executed for
more time units than the oldest active at $\widehat{S}_{i+1}+P_{i+1}$. Therefore
since $\nexists t \in [\widehat{S}_{i+1}, \widehat{S}_{i+1}+P_{i+1})$ such that at $t$
there is at least one processor available in $\sigma^{(i)}$ and no job
of $\tau_{i+1}$ is scheduled at $t$ in $\sigma^{(i+1)}$, then we have
that there are no sufficient time instants when at least one processor
is available to schedule all the jobs actived of $\tau_{i+1}$ in the
time interval $[\widehat{S}_{i+1}, \widehat{S}_{i+1}+P_{i+1})$. We obtain that the system
is not feasible, which is in contradiction with our assumption of
$\tau$ being feasible.

Consequently $\theta(\widehat{S}_{i+1})=\theta(\widehat{S}_{i+1}+P_{i+1})$, moreover by
definition of $\widehat{S}_{i+1}$ (which corresponds to an activation of
$\tau_{i+1}$) the task activations repeat from $\widehat{S}_{i+1}$ which proves
the property.
\end{proof}

\section{Exact feasibility tests} \label{sectionExactFebTest}

In the previous sections, we assumed that the execution requirement of
each task is constant while the designer knows actually only an upper
bound on the actual execution requirement, i.e., the worst case
execution time (WCET). Consequently, we have to show that our tests are
\emph{robust}, i.e., considering the scenario where all task
requirements are the maximal ones is indeed the worst case scenario,
which is not obvious upon multiprocessors because of scheduling
anomalies. More precisely, we have to show that the considered
schedulers upon the considered platforms are \emph{predictable}. Based
on this property of predictability and the periodicity results of
Section~\ref{sectionMainPer}, we provide exact feasibility tests for
the various kind schedulers and platforms considered in this work.

First of all, we introduce and formalize the notion of \emph{feasibility
interval} necessary to provide the exact feasibility tests:

\begin{Definition}[Feasibility interval]
  For any task system $\tau = \{ \tau_1, \ldots, \tau_n \}$ and any
  set of $m$ processors $\{\pi_1, \ldots, \pi_m \}$, the {\em feasibility
    interval} is a finite interval such that if no deadline is missed
  while considering only requests within this interval then no
  deadline will ever be missed.
\end{Definition}

\subsection{Preliminary results} \label{sectPremRes}


In this section, we consider the scheduling of sets of job $J \equals J_{1}, J_{2}, J_{3}\ldots$, (finite or infinite set of jobs) and without loss of generality we consider jobs in decreasing order of priorities $(J_1 > J_2 > J_{3} > \cdots$). We suppose that the execution times of each job $J_i$ can be any value in the interval $[e_i^{-}, e_i^{+}]$ and we denote by $J^{+}_i$ the job defined from job $J_i$ as follows:
$J^{+}_i \equals (r_i,e_i^{+},d_i)$. The associated execution rates of
$J^{+}_i$ are $s_{i,j}^{+} \equals s_{i,j}, \forall j$.  Similarly,
$J^{-}_i$ is the job defined from $J_i$ as follows:
$J^{-}_i=(r_i,e_i^{-},d_i)$. Similarly, the associated execution rates
of $J^{-}_i$ are $s_{i,j}^{-} \equals s_{i,j}, \forall j$. We denote
by $J^{(i)}$ the set of the first $i$ higher priority jobs. We denote
also by $J^{(i)}_{-}$ the set $\{ J^{-}_1, \ldots, J^{-}_i \}$ and by
$J^{(i)}_{+}$ the set $\{J^{+}_1, \ldots, J^{+}_i \}$. Notice that the
schedule of an ordered set of jobs using a work-conserving and
job-level fixed-priority algorithm is unique. Let $S(J)$ be the time instant at
which the lowest priority job of $J$ begins its execution in the
schedule. Similarly, let $F(J)$ be the time instant at which the
lowest priority job of $J$ completes its execution in the schedule.

\begin{Definition}[Predictable algorithms]\label{predAlg}
  A scheduling algorithm is said to be {\em predictable} if $S(J^{(i)}_{-})
  \leq S(J^{(i)}) \leq S(J^{(i)}_{+})$ and $F(J^{(i)}_{-}) \leq
  F(J^{(i)}) \leq F(J^{(i)}_{+})$, for all $1 \leq i \leq \ell$ and for all
feasible $J^{(i)}_{+}$ sets of jobs.
\end{Definition}

In~\cite{Ha} the authors showed that work-conserving job-level fixed-priority
algorithms are predictable on \emph{identical} processors. We will now extend that result by considering \emph{unrelated} platforms.

But first, we will adapt the definition availability of processors (Definition~\ref{defAvai}) to deal with the scheduling of \emph{jobs}. 

\begin{Definition}[Availability of the processors $A(J,t)$, job scheduling]\label{defAvaiJob}
For any ordered set of jobs $J$ and any set of $m$ unrelated processors
$\{\pi_1, \ldots, \pi_m \}$, we define the
  {\em availability of the processors} $A(J,t)$ of the set of jobs $J$
  at instant $t$ as the set of available processors: $A(J,t) \equals
  \{j \mid \mbox{ } \sigma_j(t)=0 \} \subseteq \{1, \ldots, m \}$, where
  $\sigma$ is the schedule of $J$.
\end{Definition}

\begin{Lemma}\label{lemmaSoon} 
  For any feasible ordered set of jobs $J$ (using the
  job-level fixed-priority and work-conserving schedule) upon an
  arbitrary set of unrelated processors $\{\pi_1, \ldots, \pi_m\}$, we
  have that $A(J^{(i)}_{+},t) \subseteq A(J^{(i)},t)$, for all $t$ and
  all $i$. That is, at any time instant the processors available in
  ${\sigma^{(i)}_{+}}$ are also available in ${\sigma^{(i)}}$. (We
  consider that the sets of jobs are ordered in the same decreasing
  order of the priorities, i.e., $J_1 > J_2 > \cdots > J_{\ell}$ and
  $J_1^{+} > J_2^{+} > \cdots > J_{\ell}^{+}$.)
\end{Lemma}

\begin{proof}
  The proof is made by induction by $\ell$ (the number of jobs).  Our
  inductive hypothesis is the following: $A(J^{(k)}_{+},t) \subseteq
  A(J^{(k)},t)$, for all $t$ and $1 \leq k \leq i$.
 
  The property is true in the base case since $A(J^{(1)}_{+},t)
  \subseteq A(J^{(1)},t)$, for all $t$. Indeed, $S(J^{(1)}) =
  S(J^{(1)}_{+})$. Moreover $J_{1}$ and $J_{1}^{+}$ are both scheduled on
  their fastest (same) processor $\pi_{n_{1,1}}$, but $J_{1}^{+}$ will
  be executed for the same or a larger amount of time than $J_{1}$.

  We will show now that $A(J^{(i+1)}_{+},t) \subseteq
  A(J^{(i+1)},t)$, for all $t$.

  Since the jobs in $J^{(i)}$ have higher priority than $J_{i+1}$, then
  the scheduling of $J_{i+1}$ will not interfere with higher priority jobs
  which are already scheduled. Similarly, $J^{+}_{i+1}$ will not
  interfere with higher priority jobs of $J^{(i)}_{+}$ which are
  already scheduled. Therefore, we may build the schedule
  $\sigma^{(i+1)}$ from $\sigma^{(i)}$, such that the jobs $J_1, J_2,
  \ldots, J_{i}$, are scheduled at the very same instants and on the
  very same processors as they were in $\sigma^{(i)}$. Similarly, we
  may build $\sigma^{(i+1)}_{+}$ from $\sigma^{(i)}_{+}$.

  Notice that $A(J^{(i+1)},t)$ will contain the same available
  processors as $A(J^{(i)},t)$ for all $t$ except the time instants at
  which $J^{(i+1)}$ is scheduled, and similarly $A(J^{(i+1)}_{+},t)$
  will contain the same available processors as $A(J^{(i)}_{+},t)$ for
  all $t$ except the time instants at which $J^{(i+1)}_{+}$ is
  scheduled. From the inductive hypothesis we have that
  $A(J^{(i)}_{+},t) \subseteq A(J^{(i)},t)$, for all $t$, and
  consequently, at any time instant $t$ we have the following
  situations:

  \begin{itemize}
  \item there is at least one eligible processor in $A(J^{(i)},t)
    \backslash A(J^{(i)}_{+},t)$ and among them the fastest processor
    is faster than those belonging to $ A(J^{(i)}_{+},t))$. Consequently,
    $J_{i+1}$ can be scheduled at time instant $t$ on faster
    processors than $J_{i+1}^{+}$.
  \item there is no eligible processor in $A(J^{(i)},t) \backslash
    A(J^{(i)}_{+},t)$. Consequently, $J_{i+1}$ can be scheduled at time
    instant $t$ on the very same processor as $J_{i+1}^{+}$.
 \end{itemize}

 Therefore, $J_{i+1}$ can be scheduled either at the very same instants
 than $J_{i+1}^{+}$ on the very same or faster processors, or may
 progress during additional time instants. Combined with the fact that
 $e_{i}\leq e_{i}^+$ the property follows for both situations.
\end{proof}

\begin{Theorem} \label{thNotWorkPred} Job-level fixed-priority algorithms are predictable on unrelated platforms.
\end{Theorem}

\begin{proof}
  For a feasible ordered set $J$ of $\ell$ jobs and a set of unrelated
  processors $\{\pi_1, \ldots, \pi_m\}$, we have to show that
  $S(J^{(i)}_{-}) \leq S(J^{(i)}) \leq S(J^{(i)}_{+})$ and
  $F(J^{(i)}_{-}) \leq F(J^{(i)}) \leq F(J^{(i)}_{+})$, for all $1
  \leq i \leq \ell$. (The sets of jobs are ordered in the same
  decreasing order of the priorities, i.e., $J_1^{-} > J_2^{-} >
  \cdots > J_{\ell}^{-}$, $J_1 > J_2 > \cdots > J_{\ell}$ and $J_1^{+}
  > J_2^{+} > \cdots > J_{\ell}^{+}$.)

  The proof is made by induction by $\ell$ (the number of jobs). We
  show the second part of each inequality, i.e. $ S(J^{(i)}) \leq
  S(J^{(i)}_{+})$ and $ F(J^{(i)}) \leq F(J^{(i)}_{+})$, for all $1
  \leq i \leq \ell$. The proof of the first part of the inequality is
  similar.

  Our inductive hypothesis is the following: $ S(J^{(k)}) \leq
  S(J^{(k)}_{+})$ and $ F(J^{(k)}) \leq F(J^{(k)}_{+})$, for all $1
  \leq k \leq i$.

  The property is true in the base case since $S(J^{(1)}) =
  S(J^{(1)}_{+})$ and $F(J^{(1)}) \leq F(J^{(1)}_{+})$.

  We will show now that $ S(J^{(i+1)}) \leq S(J^{(i+1)}_{+})$ and
  $F(J^{(i+1)}) \leq F(J^{(i+1)}_{+})$.

  Since the jobs in $J^{(i)}$ have higher priority than $J_{i+1}$ then
  the scheduling of $J_{i+1}$ will not interfere with higher priority
  jobs which are already scheduled. Similarly, $J^{+}_{i+1}$ will not
  interfere with higher priority jobs of $J^{(i)}_{+}$ which are
  already scheduled.  Therefore, we may build the schedule
  $\sigma^{(i+1)}$ from $\sigma^{(i)}$, such that the jobs $J_1, J_2,
  \ldots, J_{i}$, are scheduled at the very same instants and on the
  very same processors as they were in $\sigma^{(i)}$. Similarly, we
  may build $\sigma^{(i+1)}_{+}$ from $\sigma^{(i)}_{+}$.  The job
  $J_{i+1}$ can be scheduled only when processors, for which the
  associated execution rates are not equal to zero, are available in
  $\sigma^{(i)}$ and at those time instants $t_0 \geq r_{i+1}$ for
  which $A(J^{(i)},t)$ contains at least one eligible
  processor. Similarly, $J^{+}_{i+1}$ may be scheduled at those time
  instants $t_0^{+} \geq r_{i+1}$ for which $A(J^{(i)}_{+},t) $
  contains at least one eligible processor. By the inductive
  hypothesis we know that higher priority jobs complete sooner (or at
  the same time) consequently $t_0 \leq t_0^{+}$ and $J_{i+1}$ begins
  its execution in $\sigma^{(i+1)}$ sooner or at the same instant than
  $J_{i+1}^{+}$ in $\sigma^{(i+1)}_{+}$, i.e. $ S(J^{(i+1)}) \leq
  S(J^{(i+1)}_{+})$. It follows by Lemma~\ref{lemmaSoon} that from
  time $t_{0}$ the job $J_{i+1}$ can be scheduled at least at the very
  same instants and on the very same processors than $J_{i+1}^{+}$,
  but the job $J_{i+1}$ may also progress at the very same instants on
  faster processors (relatively to its associated set of processors)
  or during additional time instants (since we consider
  work-conserving scheduling). Consequently, $F(J^{(i+1)}) \leq
  F(J^{(i+1)}_{+})$.
\end{proof}

\subsection{Asynchronous constrained deadline systems and task-level fixed-priority schedulers}

Now we have the material to define an exact feasibility test for
asynchronous constrained deadline periodic systems.


\begin{Corollary}\label{help} 
For any preemptive task-level fixed-priority algorithm $A$ and for any asynchronous constrained deadline system $\tau$ on $m$ unrelated processors, we have that $\tau$ is $A$-feasible if and only if all deadlines are met in $[0, S_n+P)$ and if $\theta(S_{n})=\theta(S_{n}+P)$, where $S_i$ are defined inductively in Theorem~\ref{asynPer}. Moreover, for every task $\tau_i$ one only has to check the deadlines in the interval $[S_i, S_i+ \lcm \{T_j \mid j \leq i\})$.
\end{Corollary}

\begin{proof}
The Corollary~\ref{help} is a direct consequence of Theorem~\ref{asynPer} and Theorem~\ref{thNotWorkPred}, since task-level fixed-priority algorithms are job-level fixed-priority schedulers.
\end{proof}

The feasibility test given by Corollary~\ref{help} may be improved
as it was done in the uniprocessor case~\cite{Goossens2}, actually the prove remains for multiprocessor platforms since it  does not depend on the number of processors, nor on the kind of platforms but on the availability of the processors.

\begin{Theorem}[\cite{Goossens2}]\label{thLowBoundInt} Let $X_i$ be
  inductively defined by $X_n=S_n, X_i=O_i + \lfloor
  \frac{X_{i+1}-O_i}{T_i} \rfloor T_i$ $(i \in \{n-1, n-2, \ldots, 1
  \}$; we have that $\tau$ is $A$-feasible if and only if all deadlines are met in $[X_{1}, S_n+P)$ and if $\theta(S_{n})=\theta(S_{n}+P)$.
  \end{Theorem}

\subsection{Asynchronous arbitrary deadline systems and task-level fixed-priority policies}

Now we have the material to define an exact feasibility test for
asynchronous arbitrary deadline periodic systems.

\begin{Corollary}\label{helpbis} For any preemptive task-level fixed-priority algorithm $A$ and for any asynchronous arbitrary deadline system $\tau$ on $m$ unrelated
processors, we have that $\tau$ is $A$-feasible if and only if all deadlines are met in $[0, \widehat{S}_n+P)$ and if $\theta(\widehat{S}_{n})=\theta(\widehat{S}_{n}+P)$, where $\widehat{S}_i$ are defined inductively in Theorem~\ref{asynPerbis}.
\end{Corollary}

\begin{proof}
  The Corollary~\ref{helpbis} is a direct consequence of
  Theorem~\ref{asynPerbis} and Theorem~\ref{thNotWorkPred}, since task-level   fixed-priority algorithms are job-level fixed-priority schedulers.
\end{proof}

Notice that the length of our (feasibility) interval is
proportional to $P$ (the least common multiple of the
periods) which is unfortunately also the case of most
feasibility intervals for the \emph{simpler}
\emph{uni}processor scheduling problem (and for identical platforms or
simpler task models). In practice, the periods are usually
\emph{harmonics} which limits fairly the term $P$.

\subsection{{$\EDF$} scheduling of asynchronous arbitrary
  deadline systems}\label{edf}

We know by Corollary~\ref{edfAll1} that any deterministic, request-dependent and feasible
\EDF{} schedule is periodic with a period equal to $P$. Unfortunately,
from the best of our knowledge we have no upper bound on the time
instant at which the periodic part of the schedule begins. Examples
show that $O_{\max}+P$ is not such time instant for \EDF{} upon
multiprocessors (see~\cite{Braun2007Negative-Result} for instance).
Other examples, show that in some cases the periodic part of the
schedule begins after a very huge time interval (i.e., many
hyper-periods).

Based on Corollary~\ref{edfAll1} we will however define an \emph{exact} feasibility test under \EDF{} upon multiprocessors. The idea illustrated by Algorithm~\ref{algoedf} is to build the schedule (by means of simulation) and regularly check if the periodic part of the schedule is reached or not.

\begin{algorithm}
\SetKw{kwschedule}{Schedule}
\KwIn{task set $\tau$}
\KwOut{feasible}
\Begin{
\kwschedule\ (from 0) to $O_{\max}$ \;
\{The function \kwschedule stops the program and return false once a deadline is missed\}\\
$s_{1}$ := $\theta(O_{\max})$ \;
\kwschedule\ (from $O_{\max}$) to $O_{\max} + P$ \;
$s_{2}$ := $\theta(O_{\max} + P)$ \;
current-time := $O_{\max} + P$ \;
\While{$s_{1} \neq s_{2}$}{
	$s_{1}$ := $s_{2}$ \;
	\kwschedule\ (from current-time) to current-time + $P$ \;
	current-time := current-time + $P$ \;
	$s_{2}$ := $\theta$(current-time) \;
}
\Return{true}\;
}
\caption{Exact EDF-feasibility test upon multiprocessors\label{algoedf}}
\end{algorithm}

\subsection{The particular case of synchronous periodic task
  systems} \label{labelSectFebSynch}

In this section we present exact feasibility tests in the particular
case of synchronous periodic task systems. In Section
\ref{labelSectFebSynch1}, we study synchronous constrained deadline
task systems and in Section \ref{labelSectFebSynch2} synchronous
arbitrary deadline task systems.

\subsubsection{Synchronous constrained deadline task
  systems} \label{labelSectFebSynch1}

An exact feasibility test for synchronous constrained deadline systems
scheduled could be obtained directly by Theorem~\ref{thNotWorkPred}.

\begin{Corollary}\label{fixIdent2} For any deterministic, memoryless,
  job-level fixed-priority algorithm $A$ and any synchronous constrained deadline system $\tau$ on $m$ unrelated processors, we have that $\tau$ is $A$-feasible if and only if all
  deadlines are met in the interval $[0, P)$.
\end{Corollary}

\begin{proof}
  The result is a direct consequence of Theorem~\ref{synPer} and
  Theorem~\ref{thNotWorkPred}.
\end{proof}


\subsubsection{Synchronous arbitrary deadline task
  systems} \label{labelSectFebSynch2}

\begin{Corollary}\label{fixIdent2bis} For any preemptive task-level fixed-priority
  algorithm $A$ and any synchronous arbitrary deadline system $\tau$,
  $\tau$ is $A$-feasible on $m$ unrelated processors if and only if: all deadlines are met in the interval $[0, P)$, and $\theta(0) = \theta(P)$. 
\end{Corollary}

\begin{proof}
  The result is a direct consequence of Corollary~\ref{synPerArbi} and
  Theorem~\ref{thNotWorkPred}, since task-level fixed-priority schedulers are
  priority-driven.
\end{proof}

\section{Conclusion} \label{conclusion}
In this paper we studied the global scheduling of periodic task systems upon heterogeneous multiprocessor platforms. We provided exact feasibility tests based on periodicity properties. 

For any asynchronous arbitrary deadline periodic task system and any task-level fixed-priority scheduler (e.g., \RM) we characterized an upper bound in the schedule where the periodic part begins. Based on that property we provide   feasibility intervals (and consequently an exact feasibility tests) for those schedulers.

From the best of our knowledge such an interval is unknown for \EDF, a \emph{job}-level fixed-priority scheduler. Fortunately, based on a periodicity property we provide an algorithm which determine (by simulation means) where the periodicity is already started (if feasible), this algorithm provides an \emph{exact} feasibility test for \EDF{} upon heterogeneous multiprocessors.

\bibliographystyle{acm}
\bibliography{biblio.bib}

\end{document}